\documentclass[preprint]{llncs}

\usepackage{amsmath}
\usepackage{amssymb}
\usepackage{mathtools}
\usepackage{color}
\usepackage{algorithm}
\usepackage[noend]{algpseudocode}
\usepackage{listings}
\usepackage{multirow}
\usepackage{graphicx}
\usepackage{paralist}
\usepackage[title]{appendix}

\usepackage{xspace}
\newcommand{\toolname}{Z3strBV\xspace}

\begin{document}
\title{A Solver for a Theory of Strings and Bit-vectors\vspace{-0.5cm}}
\author{Sanu Subramanian\inst{1} \and Murphy Berzish\inst{1} \and \\ Yunhui Zheng\inst{2} \and Omer Tripp\inst{3} \and Vijay Ganesh\inst{1}}
\institute{University of Waterloo, Waterloo, Canada
  \and
  IBM Research, Yorktown Heights, USA
  \and
  Google, USA}

\maketitle

In this paper we present a solver for a many-sorted first-order
quantifier-free theory $T_{w,bv}$ of string equations, string length
represented as bit-vectors, and bit-vector arithmetic aimed at formal
verification, automated testing, and security analysis of C/C++
applications. Our key motivation for building such a solver is the
observation that existing string solvers are not efficient at modeling
the combined theory over strings and bit-vectors. Current approaches
either model such combination of theories by a reduction of strings to
bit-vectors and then use a bit-vector solver as a backend, or model
bit-vectors as natural numbers and use a backend solver for the
combined theory of strings and natural numbers. Both these approaches
are inefficient for different reasons. Modeling strings as bit-vectors
destroys a lot of structure inherent in string equations thus missing
opportunities for efficiently deciding such formulas, and modeling
bit-vectors as natural numbers is well known to be inefficient. Hence,
there is a clear need for a solver that models strings and bit-vectors
natively.

Our solver \toolname is such a decision procedure for the theory
$T_{w,bv}$ that combines a solvers for bit-vectors and string
equations. We demonstrate experimentally that \toolname is
significantly more efficient than a reduction of string/bit-vector
constraints to strings/natural numbers followed by a solver for
strings/natural numbers or modeling strings as
bit-vectors. Additionally, we prove decidability for the theory
$T_{w,bv}$. We also propose two optimizations, which can be adapted to
other contexts. The first accelerates convergence on a consistent
assignment of string lengths, and the second ---
dubbed \emph{library-aware SMT solving} --- fixes summaries for
built-in string functions (e.g., \texttt{strlen} in C/C++),
which \toolname consumes directly instead of analyzing the functions
from scratch each time.  Finally, we demonstrate experimentally
that \toolname is able to detect nontrivial overflows in real-world
system-level code, as confirmed against 7 security vulnerabilities
from CVE and Mozilla database.

\section{Introduction}

In recent years, constraint solvers have increasingly become the basis
of many tools for hardware verification \cite{stp}, program
analysis \cite{FSE13zheng,cav15,PISA,hampi} and automated
testing \cite{EXE,Cadar:2008:KUA:1855741.1855756,Sen:2013:JSR:2491411.2491447,Sen:2005:CCU:1081706.1081750,Godefroid:2005:DDA:1065010.1065036}. The
key idea is to model behaviors of interest of the subject system as
logical constraints, and then discharge the resulting constraints to a
SAT or SMT solver such that solutions generated by the solver serve as
test inputs for the system-under-verification.

Naturally, the ability to carry out reasoning in this fashion is
dependent on the expressive power and efficiency of the solver, which
has motivated significant effort in developing useful theories,
integrate them into solvers, and optimize their ability to solve such
rich classes of constraints. Examples include the quantifier-free (QF)
first-order theory of bit-vectors, which is effective in modeling
machine arithmetic \cite{stp,z3}; the QF theory of arrays, which
enables modeling of machine memory~\cite{stp,z3}; and the theory
of uninterpreted functions and integers to model abstractions of
program state \cite{Cadar:2008:KUA:1855741.1855756}.

\paragraph{\bf Existing Solutions, and the need for a New Solver:}
There are several powerful tools to reason about string-manipulating
code \cite{hampi,PISA,s3,cav15,CVC4-CAV14}. All these tools support
the theory of string equations, where the length function --- applied
to a string --- returns an arbitrary-precision (or unbounded) natural
number. (HAMPI~\cite{hampi} is an exception since it only deals with
bounded-length string variables.) While effective, these solvers are
not adequate for deciding the quantifier-free first-order many-sorted
theory $T_{w,bv}$ over the language $L_{w,bv}$ of string
concatenations, bit-vector-sorted string length terms, bit-vector
arithmetic, and equality predicate over string and bit-vector
terms. The reason is that modeling bit-vectors as natural numbers
complicates reasoning about potential overflows (or underflows), which
is an artifact of the fixed-precision machine (bit-vector)
representation of numeric values. Precise modeling of arithmetic
overflow/underflow is a key reason to model numeric values in terms of
bit-vectors and not integers, and the motivation for a large body of
work on bit-vector solvers\cite{DBLP:conf/cav/GaneshD07,z3}.

Another approach to solving $L_{w,bv}$-formulas is to represent
strings as bit-vectors. In fact many symbolic execution engines like
KLEE \cite{Cadar:2008:KUA:1855741.1855756} and
S2E \cite{Chipounov:2011:SPI:1950365.1950396} perform reasoning at
this level. They collect constraints as bit-vectors by solving branch
conditions using STP \cite{DBLP:conf/cav/GaneshD07} or
Z3 \cite{z3}. However, these engines perform poorly on programs that
make heavy use of string functions, as the low-level bit-vector
representation of program data fails to efficiently capture the
high-level semantics of the string data
type \cite{Chipounov:2011:SPV}.

As this brief survey highlights, currently there is a disturbing
gap. String solvers that are based on the theory of string equations
and linear natural number arithmetic are inadequate for solving
strings and bit-vectors given their limited ability to model overflow,
underflow, bit-wise operations, and pointer casting. At the same time,
there is lot of empirical evidence that bit-vector solvers are not
able to perform direct reasoning on strings
efficiently~\cite{anathesis}. Furthermore, they cannot handle
unbounded strings.

Hence, we were motivated to build \toolname{}, which solves
$L_{w,bv}$-formulas by treating strings and bit-vectors natively.
We did so by combining a solver for strings, augmented with
a bit-vector-sorted length function, and a solver for bit-vectors
within the Z3 SMT solver combination framework.

\paragraph{\bf Motivating Example:} This paper follows the same general motivation as that of a typical
SMT solver design and implementation. Specifically, we address the
%
increasing need of {\it efficient} solver supports to reason about security errors due to
improper string manipulations, which occur frequently in C/C++
system-level code \cite{DBLP:conf/esem/QuR11}. Improper handling of
string values carries serious ramifications, including crashes,
unintended program behaviors, and exposure to security
threats \cite{OPLDI09,OACSAC15,OISSTA13,OISSTA14}.
A constraint solver for such analyses must be able to model not only string values
but also machine-level constraints in bit-vector arithmetic, and in
particular the potential for arithmetic overflow.

\lstset{basicstyle=\scriptsize,breaklines=true}
\vspace{-0.25in}
\begin{figure}
\centering{
\begin{lstlisting}[language=C,numbers=left,xleftmargin=.1\textwidth]
bool check_login(char* username, char* password) {
  if (! validate_password(password)) {
     invalid_login_attempt(); exit(-1);
  }
  const char *salt = get_salt8(username);
  unsigned short len = strlen(password) + strlen(salt) + 1;
  if (len > 32) {
    invalid_login_attempt(); exit(-1);
  }
  char *saltedpw = (char*)malloc(len);
  strcpy(saltedpw, password);
  strcat(saltedpw, salt);
  ...
}
\end{lstlisting}
\vspace{-0.1in}
\caption{A login function with an integer overflow vulnerability}
\label{Fi:motivating}
}
\end{figure}
\vspace{-0.15in}

As an illustration, we consider an example
(Figure~\ref{Fi:motivating}) inspired by real-world vulnerabilities we
analyzed. It uses a combination of string manipulations, the string
length function, and bit-vector arithmetic, and is interesting to
analyze from the perspective of this theory combination.  The {\tt
validate\_password} at line $2$ abstracts a method that utilizes
string operators (e.g., {\tt strstr} and {\tt strcmp}) and makes sure
the input is safe, for example, that it does not contain any
non-printable characters, but does not perform any checking of the
input length.  The {\tt get\_salt8} method maps a username to a
pseudorandom 8-character string, which is concatenated with the
password to strengthen it against certain types of password cracking
attacks.  Although the input length is indirectly checked at line $7$,
there is still the threat of an overflow.  If the input consists of
$65535 - 8 = 65527$ characters, an overflow occurs when the variable
{\tt len}, which is of type {\tt unsigned short} and thus ranges over
$[0, 2^{16}-1 = 65535]$, is assigned {\tt strlen(password) +
strlen(salt) + 1}. This leads to the allocation of a buffer of size 0,
and consequently to heap corruption due to copying of {\tt password}
into the empty buffer.

Analyses of such programs based on symbolic executions with a
traditional string solver (without bit-vector support) or a purely
bit-vector solver can be {\it highly inefficient} \cite{anathesis}.  A
source of complexity is the handling of integer modulo operations
required to simulate the fixed-precision calculations and overflows.
Additionally, other string operations listed require the use of a
solver that supports strings as well as bit-vectors, as simple bounds
checking alone is insufficient to capture the relevant semantics of
this method.

\vspace{-0.1in}

\paragraph{\bf Summary of Contributions:}
The key novelty and contribution of this paper is a solver algorithm
for a combined theory of string equations, string length modelled using
bit-vectors, and bit-vector arithmetic, its theoretical underpinnings,
implementation, and evaluation over several sets of benchmarks. The
contributions, in more detail, are:

\vspace{-0.05in}

\begin{enumerate}

\item {\bf (Formal Characterization)}
We begin with a formal characterization of a theory of string
equations, string length as bit-vectors, and linear arithmetic over
bit-vectors, and provide a (constructive) proof of decidability for
this theory.  In particular, we want to stress that the decidability
of such a theory is non-trivial. At first glance it may seem that all
models for this theory have finite universes since bit-vector
arithmetic has a finite universe. However, this is misleading since
the search space over all possible strings remains infinite even when
string length is a fixed-width bit-vector due to the wraparound
semantics of bit-vector arithmetic. (For example, for any fixed length
$N$, there are infinitely many strings whose length is some constant k
modulo $N$.)

\item {\bf (Solver Algorithm)}
We then specify a practical solver algorithm for this theory that is
efficient for a large class of verification, testing, analysis and
security applications. Additionally, we formally prove that our
solving algorithm is sound.

\item {\bf (Enhancements)}
We propose two optimizations to our solver algorithm, whose
applicability reaches beyond the string/bit-vector scope. We introduce
the notion of \emph{library-aware solving}, whereby the solver reasons
about certain C/C++ string library functions natively at the contract
(or summary) level, rather than having to (re)analyze their actual
code (and corresponding constraints) each time the symbolic analysis
encounters them.  We also propose a ``binary search'' heuristic, which
allows fast convergence on consistent string lengths across the string
and bit-vector solvers. This heuristic can be of value in other theory
combinations such as the ones over theories of string and natural
number.

\item {\bf (Implementation and Evaluation)}
Finally, we describe the implementation of our solver, \toolname,
which is an extension of the Z3str2 string solver.  We present
experimental validation for the viability and significance of our
contributions, including in particular (i) the ability to detect
overflows in real-world systems using our solver, as confirmed via
reproduction of several security vulnerabilities from the CVE
vulnerability database, and (ii) the significance of the two
optimizations outlined above.
\end{enumerate}

\section{Syntax and Semantics}

\subsection{The Syntax of the Language $L_{w, bv}$}

We define the sorts and constant, function, and predicate symbols of
the countable first-order many-sorted language $L_{w, bv}$.

\noindent \textbf{Sorts:} The language is many-sorted with a string
sort $str$ and a bit-vector sort $bv$.  The language is parametric in
$k$, the width of bit-vector terms (in number of bits). The Boolean
sort $Bool$ is standard. When necessary, we write the sort of an
$L_{w,bv}$-term $t$ explicitly as $t:sort$.

\noindent \textbf{Finite Alphabet:} We assume a finite alphabet
$\Sigma$ of characters over which all strings are defined.

\noindent \textbf{String and Bit-vector Constants}: We define a
disjoint two-sorted set of constants $Con = Con_{str} \cup
Con_{bv}$. The set $Con_{str}$ is a subset of $\Sigma^{*}$, the set of
all finite-length string constants over the finite alphabet $\Sigma$.
Elements of $Con_{str}$ will be referred to as \textit{string
  constants}, or simply \textit{strings}.  $\epsilon$ denotes the
empty string. Elements of $Con_{bv}$ are binary constants over $k$
digits.  As necessary, we may subscript bit-vector constants by $bv$
to indicate that their sort is ``bit-vector''.

\noindent \textbf{String and Bit-vector Variables:} We fix a disjoint
two-sorted set of variables $var = var_{str} \cup var_{bv}$.
$var_{str}$ consists of string variables, denoted $X, Y, \hdots$ that
range over string constants, and $var_{bv}$ consists of bit-vector
variables, denoted $a, b, \hdots$ that range over bit-vectors.

\noindent \textbf{String Function Symbols:} The string function
symbols include the concatenation operator $\cdot : str \times str \to
str$ and the length function $strlen_{bv} : str \to bv$.

\noindent \textbf{Bit-vector Arithmetic Function Symbols:} The
bit-vector function symbols include binary $k$-bit addition (with
overflow) $+: bv \times bv \to bv$. Following standard practice in
mathematical logic literature, we allow multiplication by constants as
a shorthand for repeated addition.

\noindent \textbf{String Predicate Symbols:} The predicate symbols
over string terms include equality and inequality: $=,\neq : str
\times str \to Bool$.

\noindent \textbf{Bit-vector Predicate Symbols:} The predicate symbols
over bit-vector terms include $=$, $\ne$, $<$, $\le$, $>$, $\ge$ (with
their natural meaning), all of which have signature $bv \times bv \to
Bool$.

\subsection{Terms and Formulas of $L_{w, bv}$}

\noindent{\textbf{Terms:}} $L_{w,bv}$-terms may be of string or
bit-vector sort.  A string term $t_{str}$ is inductively defined as an
element of $var_{str}$, an element of $Con_{str}$, or a concatenation
of string terms.  A bit-vector term $t_{bv}$ is inductively defined as
an element of $var_{bv}$, an element of $Con_{bv}$, the length
function applied to a string term, a constant multiple of a length
term, or a sum of length terms.  (For convenience we may write the
concatenation and addition operators as $n$-ary functions, even though
they are defined to be binary operators.)

\noindent{\textbf{Atomic Formulas:}} The two types of atomic formulas
are (1) word equations ($A_{w}$) and (2) inequalities over bit-vector
terms ($A_{bv}$).

\noindent{\textbf{QF Formulas:}} We use the term ``QF formula'' to
refer to any Boolean combination of atomic formulas, where each free
variable is implicitly existentially quantified and no explicit
quantifiers may be written in the formula.

\noindent{\bf Formulas and Prenex Normal Form:} $L_{w,bv}$-formulas
are defined inductively over atomic formulas. We assume that formulas
are always represented in prenex normal form (i.e., a block of
quantifiers followed by a QF formula).

\noindent{\bf Free and Bound Variables, and Sentences:} We say that a
variable under a quantifier in a formula $\phi$ is bound. Otherwise we
refer to variables as free. A formula with no free variables is called
a sentence.

\subsection{Semantics and Canonical Model over the Language $L_{w,bv}$}

We fix a string alphabet $\Sigma$ and a bit-vector width $k$. Given
$L_{w, bv}$-formula $\phi$, an \textit{assignment} for $\phi$
w.r.t. $\Sigma$ is a map from the set of free variables in $\phi$ to
$Con_{str} \cup Con_{bv}$, where string (\emph{resp.} bit-vector)
variables are mapped to string (\emph{resp.}  bit-vector) constants.
Given such an assignment, $\phi$ can be interpreted as an assertion
about $Con_{str}$ and $Con_{bv}$.  If this assertion is true, then we
say that $\phi$ itself is \textit{true} under the assignment.  If
there is some assignment s.t. $\phi$ is true, then $\phi$ is
\textit{satisfiable}.  If no such assignment exists, then $\phi$ is
\textit{unsatisfiable}.

For simplicity we omit most of the description of the canonical model
of this theory, choosing to use the intuitive combination of
well-known models for word equations and bit-vectors.  We provide,
however, semantics for the $strlen_{bv}$ function, since it is not a
``standard'' symbol of either separate theory.

\noindent{\textbf{Semantics of the $strlen_{bv}$ Function:}} For a
string term $w$, $strlen_{bv}(w)$ denotes
an unsigned, fixed-precision bit-vector representation of
the ``precise'' integer length of $w$, truncating the arbitrary-precision bit-vector representation
of that integer to its lowest $k$ bits, for fixed bit-vector width $k$.
The bit-vector addition and
(constant) multiplication operators produce a result of the same width
as the input terms and treat both operands as though they represent
unsigned integers. Of particular note is that both of these operators
have the potential to overflow (that is, to produce a result that is
smaller than either operand).  This is a consequence of the fixed
precision of bit-vectors.
Furthermore, the $strlen_{bv}$ function itself may also ``overflow'',
because it is a fixed-width representation of an arbitrary-precision integer.
More precisely, bit-vector arithmetic has the semantics of integer arithmetic modulo $2^k$,
and the value represented by $strlen_{bv}(w)$ is the bit-vector representation of the
value in the field of integers modulo $2^k$
that is congruent to the ``precise'' integer length of $w$.

For example, if the number of bits used to represent bit-vectors is $k=3$,
a string of precise length 1 and another string of precise length 9
both have bit-vector width of ``001''. Although the complete bit-vector representation of 9
as an arbitrary-precision bit-vector would be ``1001'', the semantics of $strlen_{bv}$ specify
that all but the $k=3$ lowest bits are omitted.

Note that the search space with respect to strings is (countably) infinite
despite the fixed-width representation of string lengths as bit-vectors.
This is because the bit-vector length of a string is only a view of its precise length,
i.e. the integer number of characters in the string. This integer length may be arbitrarily finitely large.
The semantics of fixed-width integer overflow is, in essence,
applied to the integer length in order to obtain the bit-vector length.
In fact, there are infinitely many strings that have the same bit-vector length.
For example, if eight bits are used to represent string length,
strings of length 0, 256, 512, ... would all appear to have a bit-vector length of ``00000000''.

\section{Decidability of QF String Equations, String Length, and Bit-Vector Constraints}

In this section, we prove the decidability of the theory $T_{w,bv}$ of QF word
equations and bit-vectors.  Towards this goal, we first establish a
conversion from bit-vector constraints to regular languages. For
regular expressions (regexes), we use the following standard notation:
$AB$ denotes the concatenation of regular languages $A$ and $B$.
$A|B$ denotes the alternation (or union) of regular languages $A$ and
$B$.  $A^{*}$ denotes the Kleene closure of regular language $A$
(i.e., 0 or more occurrences of a string in $A$).  For a finite
alphabet $\Sigma = \{a_1, a_2, \hdots, a_l\}$, $\left[ a_1 -
a_l \right]$ denotes the union of regex $a_1 | a_2 | \hdots | a_l$.
Finally, $A^{i}$, for nonzero integer constants $i$ and regex $A$,
denotes the expression $A A \hdots A$, where the term $A$ appears $i$
times in total.

\begin{lemma} \label{lem:bv2regex}
  Let $k$ be the width of all bit-vector terms.  Suppose we have a
  bit-vector formula of the form $len_{bv}(X) = C$, where $X$ is a
  string variable and $C$ is a bit-vector constant of width $k$.  Let
  $i_{C}$ be the integer representation of the constant $C$,
  interpreting $C$ as an unsigned integer.  Then the set $M(X)$ of all
  strings satisfying this constraint is equal to the language $L$
  described by the regular expression $(\left[a_1 -
  a_l\right]^{2^{k}})^{*} \left[a_1 - a_l\right]^{i_{C}}$.
  (Refer to Appendix~\ref{proof:lem:bv2regex} for proof.)
\end{lemma}

\noindent{\bf Proof Idea for the Decidability Theorem~\ref{thm:strbvdecidable}:}
Intuitively, the decision procedure proceeds as follows. The crux of
the proof is to convert bit-vector constraints into regular languages
(represented in terms of regexes) relying on the lemma mentioned
above, and correctly capture overflow/underflow behavior. In order to
capture the semantics of unsigned overflow, each regex we generate has
two parts: The first part, under the Kleene star, matches strings of
length a multiple of $2^{k}$ that cause an $k$-bit bit-vector to
overflow and wrap around to the original value; the second part
matches strings of constant length $i_C$, corresponding to the part of
the string that is ``visible'' as the bit-vector length
representation. By solving the bit-vector fragment of the equation
first we can generate all of finitely many possible solutions, and
therefore check each of finitely many assignments to the bit-vector
length terms. For each bit-vector solution, we solve the word-equation
fragment separately under regular-language constraints, which
guarantee that only strings that have the expected bit-vector length
representation will be allowed as solutions. It is easy to see that
this algorithm is sound, complete, and terminating, given a decision
procedure for word equations and regex.

\begin{theorem} \label{thm:strbvdecidable}
  The satisfiability problem for the QF theory of word equations and
  bit-vectors is decidable. (Refer to Appendix~\ref{proof:thm:strbvdecidable} for proof.)
\end{theorem}

It may appear that the decidability result is trivial as the domain of bit-vectors is finite for fixed
width $k$, and therefore the formula could be solved by trying all $2^k$ possible assignments for each bit-vector length term
and all finitely many strings whose length is equal to each given bit-vector length. However, this is incorrect,
as the bit-vector length of a string is only a representation of its ``true length''. As the semantics of bit-vector arithmetic
specify that overflow is possible under this interpretation, strings of integer length 1, 5, 9, etc. -- indeed, infinitely many strings --
will all satisfy a constraint asserting that the 2-bit bit-vector length of a string term is 1. 
Therefore, it is not sufficient to search over, for example,
only the space of strings with length between 0 and $2^{k} - 1$. Hence the decidability of this theory
is non-trivial, and this motivates the need for a stronger argument, such as given in Theorem~\ref{thm:strbvdecidable}.

\section{\toolname Solver Algorithm}

The solver algorithm that we have designed differs from the decision
procedure described in the proof for Theorem~\ref{thm:strbvdecidable}.
There are two main reasons for that. First, that decision procedure is
completely impractical to implement as an efficient constraint solver.
Second, and related to the previous point, that decision procedure
does not leverage existing solving infrastructure, which implies
considerably more engineering effort. Our solving algorithm, on the
other hand, builds on the Z3str2 technique to solve word
equations \cite{FSE13zheng,cav15}, including in particular boundary
labels, word-equation splits, label arrangements and detection of
overlapping variables.  In contrast with Z3str2, however, (i) we
perform different reasoning about length constraints derived from the
word equations, and (ii) we have a different search-space pruning
strategy to reach consistent length assignments.

\subsection{Pseudocode Description}

The main procedure of the \toolname solver, which is similar to the
Z3str2 procedure \cite{FSE13zheng,cav15}, is summarized as
Algorithm~\ref{alg:highLevel}. It takes as input sets $\mathcal{Q}_w$
of word equations and $\mathcal{Q}_l$ of bit-vector (length)
constraints, and its output is either SAT or UNSAT or UNKNOWN.
UNKNOWN means that the algorithm has encountered overlapping
arrangements and pruned those arrangements (thereby potentially
missing a SAT solution), and a SAT solution could not be found in
remaining parts of the search space.

\algrenewcommand\alglinenumber[1]{\scriptsize #1:}
\begin{algorithm}
	{\algrenewcommand\algorithmicindent{1.5em}
		\caption{High-level description of the \toolname main algorithm.}
		\label{alg:highLevel}
		\begin{algorithmic}[1]
			{\scriptsize
				\Statex \textbf{Input:} sets $\mathcal{Q}_w$ of word equations and
				$\mathcal{Q}_l$ of bit-vector (length) constraints
				\Statex \textbf{Output:} SAT / UNSAT / UNKNOWN 
				
				\Procedure{solveStringConstraint}{$\mathcal{Q}_w$,$\mathcal{Q}_l$}
				
				\If{all equations in $\mathcal{Q}_w$ are in solved form}
				\If{$\mathcal{Q}_w$ is UNSAT or $\mathcal{Q}_l$ is UNSAT}  
				\State \Return UNSAT\label{alg:unsat1}
				\EndIf
				
				\If{$\mathcal{Q}_w$ and  $\mathcal{Q}_l$ are SAT and mutually consistent} 
				\State \Return SAT
				\EndIf
				\EndIf

				\State $\mathcal{Q}_a$ $\longleftarrow$ Convert $\mathcal{Q}_w$ into equisatisfiable DNF formula\label{alg:conv1} 
				
				\ForAll{disjunct $D$ in $\mathcal{Q}_a$}
				\State $\mathbb{A}$ $\longleftarrow$ all possible arrangements of equations in $D$\label{alg:conv2_s}
				\ForAll{arrangement $A$ in  $\mathbb{A}$}
				\State $l_A$ $\longleftarrow$ length constraints implied by $A$
				\If{$l_A$ is inconsistent with $\mathcal{Q}_l$} 
				\State $\mathbb{A}$ $\longleftarrow$ $\mathbb{A} \setminus \{\ A\ \}$
				\EndIf
				\EndFor\label{alg:conv2_e}
				\ForAll{string variable $s$ incident in $D$}\label{alg:merge_s}
				\State $G(s)$ $\longleftarrow$ merge per-equation arrangements involving $s$
				\EndFor\label{alg:merge_e}
				\ForAll{merged arrangements $a \in G(s)$ with no overlaps}
				\State $\mathcal{Q}'_w$ $\longleftarrow$ refine variables in $\mathcal{Q}_w$ per $a$\label{alg:conv3}
				\State $\mathcal{Q}'_l$ $\longleftarrow$ update length constraints $\mathcal{Q}_l$ per $\mathcal{Q}'_w$
				\State $r$ $\longleftarrow$ \textsc{SolveStringConstraint}($\mathcal{Q}^{'}_{w}$,$\mathcal{Q}^{'}_{l}$)
				\If {$r$=SAT}
				\State \Return SAT
				\EndIf
				\EndFor
				\EndFor
			
				\If{overlapping variables detected at any stage}
				\State \Return UNKNOWN
				\Else
				\State \Return UNSAT\label{alg:unsat2}
				\EndIf 
				
				\EndProcedure
			}
		\end{algorithmic}
	}
\end{algorithm}

The input to the procedure is a conjunction of constraints.
Any higher-level Boolean structure is handled by the SMT core solver,
typically a SAT solver.
The first part of the procedure (lines 2-9) check whether (i) either
$\mathcal{Q}_w$ or $\mathcal{Q}_l$ is UNSAT or (ii) both are SAT and
the solutions are consistent with each other. If neither of these
cases applies, then arrangements that are inconsistent with the length
constraints are pruned (lines 12-21). Finally, the surviving
arrangements $G(s)$ guide refinement of the word equations
$\mathcal{Q}_w$, and so also of the length constraints
$\mathcal{Q}_l$, and for each $G(s)$ the solving loop is repeated for
the resulting sets $\mathcal{Q}'_w$ and $\mathcal{Q}'_l$ (lines
22-29). A SAT answer leads to a SAT result for the entire
procedure. If no solution is found, but overlapping variables have
been detected at some point, then the procedure returns UNKNOWN.
(Note that all current practical string solvers
suffer from both incompleteness and potential non-termination.)

Notice that during the solving process, the string plug-in
(potentially) derives additional length constraints incrementally
(line 14). These are discharged to the bit-vector solver on demand,
and are checked for consistency with all the existing length
constraints (both input length constraints and constraints added
previously during solving).

More generally, during the solving process the string and bit-vector
solvers each generate new assertions in the other domain. Inside the
string theory, candidate arrangements are constrained by the
assertions on string lengths, which are provided by the bit-vector
theory. In the other direction, the string solver derives new length
assertions as it progresses in exploring new arrangements. These
assertions are provided to the bit-vector theory to prune the search
space.

\noindent {\bf Basic Length Rules:} Given strings $X,Y,Z,W,\ldots$, we express their respective lengths
$l_X,l_Y,l_Z,\ldots$ as
$strlen\_bv(X,n),strlen\_bv(Y,n),strlen\_bv(Y,n),\ldots$ respectively
in the constraint system, where $n$ is the bit-vector sort. The empty
string is denoted by $\epsilon$. Two rules govern the reasoning
process: (i) $X=Y \implies l_{X} = l_{Y}$ and (ii) $W=X \cdot Y \cdot
Z \cdot \ldots \implies l_{W} = l_{X} + l_{Y} + l_{Z}+\cdot \ldots$.

As an example, consider the word equation $X\cdot Y = M\cdot N$, where
$X,Y,M,N$ are nonempty string variables. The are three possible
arrangements \cite{cav15}, as shown below on the left, where $T_1$ and
$T_2$ are temporary string variables. The respective length
assertions, derived from these three arrangements, are listed on the
right.
\begin{center}
\begin{tabular}{lcl}
$(X = M \cdot T_1) \wedge (N = T_1 \cdot Y)$ & $\quad\quad$ &
$(l_X = l_M + l_{T_1}) \wedge (l_N = l_{T_1} + l_Y)$ \\	
$(X = M) \wedge (N = Y)$ & &
$(l_X = l_M) \wedge (l_N = l_Y)$ \\
$(M = X \cdot T_2) \wedge (Y = T_2 \cdot N)$ & &
$(l_M = l_X + l_{T_2}) \wedge (l_Y = l_{T_2} + l_N)$
\end{tabular}
\end{center}

The proof of soundness of Algorithm~\ref{alg:highLevel} is presented
in Appendix~\ref{proof:alg:highLevel}.

We conclude by noting that the Z3str2 algorithm for string constraints is terminating, as shown in~\cite{cav15},
and that it is easy to see that integrating the theory of bit-vectors with this algorithm preserves this property,
as the space of bit-vector models is finite.

\subsection{Binary Search Heuristic} 

As explained above, length assertions are added to the Z3 core, then
processed using the bit-vector theory. For efficiency, we have
developed a binary-search-based heuristic to fix a value for the
length variables in the bit-vector theory. To illustrate the
heuristic, and the need for it, we consider the example
$$
	"a" \cdot X=Y \cdot "b" \bigwedge bv8000[16]<strlen\_bv(X,16)<bv9000[16]
$$

where $bv8000[16]$ ($bv9000[16]$) denotes the constant 8000
(9000). The constraint $"a" \cdot X=Y \cdot "b"$ is discharged to the
string theory, whereas $bv8000[16]<strlen\_bv(X,16)<bv9000[16]$ is
discharged to the bit-vector solver. For the string constraint, a
(non-overlapping) solution is
$$
	X="b" \quad Y="a" \quad strlen\_bv(X,16)=strlen\_bv(Y,16)=bv1[16]
$$
but this solution is in conflict with the bit-vector constraints. Thus, the (overlapping) arrangement $X=T \cdot "b" \quad Y="a" \cdot T$
is explored, which leads to length constraints
$$
\begin{array}{lcr}
strlen\_bv(v,16)=strlen\_bv(T,16)+bv1[16] & \colon & v \in \{ X,Y \} \\
\multicolumn{3}{c}{strlen\_bv(T,16)>bv0[16]}
\end{array}
$$

Now the need arises to find consistent lengths for $X,Y,T$. Iterating
all possibilities one by one, and checking these possibilities against
the bit-vector theory, is slow and expensive. Instead, binary search
is utilized to fix lower and upper bounds for candidate lengths.

The first choice is for a lower bound of 0 and an upper bound of
$2^{16}$ (where $16$ is the bit-vector width, as indicated
above). This leads to the first candidate being
$strlen(X,16)=bv32767[16]$ (where $32767=2^{15}-1$). This fails, leading
to an update to the upper bound to be $2^{15}$, and consequently the next
guess is $strlen(X,16)=bv16383[16]$. This too falls outside the range
$(8000,9000)$, and so the upper bound is updated again, this time
becoming $2^{14}$, and the next guess is
$strlen(X,16)=bv8191[16]$. This guess is successful, and so within 3
(rather than 8000) steps the search process converges on the following
consistent length assignments:
$$
\begin{array}{lcr}
l_v=strlen(v,16)=bv8191[16] & \colon & v \in \{ X,Y\} \\
\multicolumn{3}{c}{l_T=strlen(T,16)=bv8190[16]}
\end{array}
$$

As the example highlights, in spite of the tight interaction between
the string and bit-vector theories, large values for length
constraints are handled poorly by default, since the process of
converging on consistent string lengths is linearly proportional to
those values. Pleasingly, bit-vectors, expressing a finite range of
values, enable safe lower and upper bounds. More concretely, given a
bit-vector of width $n$, the value of the length variable is in the
range $[0,2^n-1]$. Our heuristic iteratively adds length assertions to
the bit-vector theory following a binary-search pattern until
convergence on consistent length assignments. This process is both
sound and efficient.

\subsection{Library-aware Solving Heuristic}

The concept of library-aware SMT solving is simple. The basic idea is to provide native SMT solver support for
a class of library functions $f$ in popular programming languages like
C/C++ or Java, such that (i) $f$ is commonly used by programmers, (ii) uses of $f$ are a frequent source of errors (due to programmer mistakes), and (iii) symbolic analysis of $f$ is expensive due to the many paths it defines.

More precisely, by library-aware SMT solvers we mean that the logic of
traditional SMT solvers is extended with declarative summaries of
functions such as {\tt strlen} or {\tt strcpy}, expressed as global invariants
over all behaviors of such functions. The merit of declaratively
modeling such functions is that, unlike the real code implementing
these functions, the summary is free of downstream paths to explore. 
Instead, the function is modeled as a set of logical constraints, thereby offsetting the path explosion problem.

Observe, importantly, that library-aware SMT solving is complementary to summary-based symbolic
execution. To fully exploit library-aware SMT solving, one has to
modify the symbolic execution engine as well to generate summaries or
invariants upon encountering library functions. While summary-based
symbolic execution has been studied (e.g. as part of the S-Looper tool \cite{slooper}), we are not aware of any previous work where SMT
solvers directly support programming-language library functions
declaratively as part of their logic.
One recent application of a similar concept is discussed in~\cite{Jeon2016:SymExec}, where
models of design patterns are abstracted into a symbolic execution engine; being able to perform a similar analysis
at the level of individual library methods as part of a library-aware SMT solver can be very useful to enhance
library-aware symbolic execution such as demonstrated in that work.
We intend to explore this idea further in the future to broaden its applicability beyond the current context.
Furthermore, capturing program semantics precisely and concisely in a symbolic summary mandates
integration between strings (for conciseness) and bit-vectors (for modelling overflow and precise bit-level operations).
This further motivates the connection to and importance of a native solver for strings and bit-vectors.

\section{Experimental Results}

In this section, we describe our evaluation of \toolname.  The
experiments were performed on a MacBook computer, running OS X
Yosemite, with a 2.0GHz Intel Core i7 CPU and 8GB of RAM.  We have
made the \toolname code, as well as the experimental artifacts,
publicly available \cite{toolURL}.

\vspace{-0.15in}
\subsection{Experiment I: Buffer Overflow Detection}

To validate our ability to detect buffer overflows using \toolname, we
searched for such vulnerabilities in the CVE database \cite{CVEDB}. We
selected 7 cases, and specified the vulnerable code in each case as
two semantically equivalent sets of constraints --- in the
string/natural number theory and in the string/bit-vector theory ---
to compare between \toolname and Z3str2. The Z3str2 tool is one of the
most efficient implementations of the string/natural number theory.
The solvers only differ in whether string length is modelled as an integer or as a bit-vector.
We set the solver timeout for each test case at 1
hour. Figure~\ref{tab:cve} presents the results.
\toolname is able to detect all vulnerabilities, and further generate 
corresponding input values that expose/reproduce the
vulnerability. Z3str2, by contrast, provides limited support for
arithmetic overflow/underflow. Unfortunately, correctly modeling
overflow/underflow using linear arithmetic over natural numbers is
inefficient, and thus it fails within the prescribed time budget of 1
hour. Without the ability to perform overflow modelling,
Z3str2 cannot detect overflow bugs at all, since arbitrary-precision integers cannot overflow.
This experiment, therefore, shows that \toolname can find bugs that Z3str2 does not detect (due to timeouts).

\begin{figure}
	\begin{minipage}{0.5\textwidth}
		\caption{CVE Buffer Overflow Detection and Exploit Synthesis (See \cite{SanuThesis,toolURL} for details.)}
		\label{tab:cve}
		\begin{footnotesize}
			\centering
			\begin{tabular}{|c|r|c|}
				\hline
				Vulnerability
				&  \multicolumn{1}{|c|}{\toolname}   
				&  \multicolumn{1}{|c|}{Z3str2} \\ 
				\hline
				CVE-2015-3824  &   0.079s  & TO (1h) \\
				\hline
				CVE-2015-3826  &   0.108s  & TO (1h)\\
				\hline
				CVE-2009-0585  &   0.031s  & TO (1h)\\
				\hline
				CVE-2009-2463  &   0.279s  & TO (1h)\\
				\hline
				CVE-2002-0639  &   0.116s  & TO (1h)\\
				\hline
				CVE-2005-0180  &   0.029s  & TO (1h)\\
				\hline
				FreeBSD Bugzilla  &   \multirow{2}{*}{0.038s}  & \multirow{2}{*}{TO (1h)} \\
				\#137484 & & \\
				\hline
			\end{tabular}
		\end{footnotesize}
	\end{minipage}
	\hfill
	\begin{minipage}{0.5\textwidth}
		\caption{Vulnerability Detection using KLEE and Library-aware Solving.}
		\label{tab:libraryAware}
		\begin{footnotesize}
			\centering
			\begin{tabular}{|r|r|r|r|}
				\hline
				\multirow{2}{*}{\centering Prec.}
				&  \multicolumn{2}{|c|}{Lib-aware Solving} 
				&  \multirow{2}{0.2\columnwidth}{\centering KLEE}\\ 
				\cline{2-3}
				&  \toolname  & Z3str2 & \\
				\hline
				8-bit   & 0.507s & 167s & 300s \\ 
				\hline
				16-bit  & 0.270s & TO (7200s)   & TO (7200s) \\
				\hline
			\end{tabular}
		\end{footnotesize}
	\end{minipage}
\end{figure}

\subsection{Experiment II: Library-aware SMT Solving}

\begin{figure}[ht]
    \vspace{-0.25in}
    \centering
    \includegraphics[width=0.9\columnwidth]{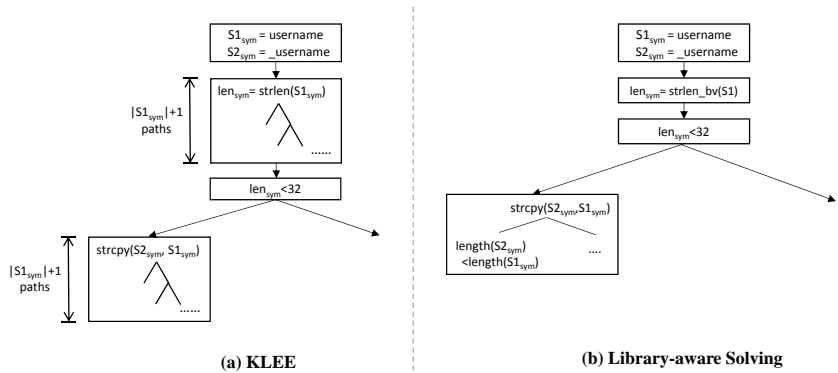}
    \label{fig:klee_libraryAware}
    \vspace{-0.15in}
    \caption{Comparison between KLEE and Library-aware Solving}
    \vspace{-0.15in}
\end{figure}

We evaluated the library-aware solving heuristic atop the example shown in Fig.~\ref{Fi:motivating} by applying both our technique and KLEE, a state-of-the-art symbolic execution engine, to this code. The goal was to detect the heap corruption threat in that code.
We faithfully encoded the program
snippet in \textsf{check\_login()} as string/bit-vector
constraints. We then checked whether the buffer pointed-to
by \textsf{\_username} is susceptible to overflow.  

Notice that \textsf{len} is an \textsf{unsigned short} variable, and thus
ranges from $0$ to $2^{16}-1 (65,535)$. As it represents the buffer
size, it determines the number of concrete execution paths KLEE has to
enumerate, as well as the search space for library-aware
solving. By contrast, the constraints generated by
library-aware SMT solver declaratively model {\tt strlen} as part of
the SMT solver logic. 

To characterize performance trends,
we consider two
different precision settings for string length: \textsf{8-bit}
and \textsf{16-bit}. We used 120 minutes as the timeout value. There was no need to go beyond 16 bits since KLEE was already
significantly slower at 16 bits relative to the library-aware SMT
solver. Note that KLEE is slow because it has to explore a large number of
paths, and not because the individual path constraints are difficult to
solve.

The results are provided in Figure~\ref{tab:libraryAware}. Under both precision settings, KLEE is
consistently and significantly slower than the library-aware solving
technique.  In particular, if we represent numeric values using 16
bits, then KLEE is not able to identify the problem in 120 minutes,
while \toolname can solve the problem in 0.27 seconds. 

The benefit thanks to library-aware solving is clear.  The analysis is as follows. Suppose both \textsf{username}
and \textsf{\_username} are symbolic string variables.  In
Figure \ref{fig:klee_libraryAware}(a), as KLEE forks a new state for
each character, an invocation of \texttt{strlen}
on a symbolic string $S1_{sym}$ of size $|S1_{sym}|$ will generate and
check $|S1_{sym}| + 1$ path constraints (one per each possible length value
between $0$ and $|S1_{sym}|$).  In
Figure \ref{fig:klee_libraryAware}(b), in contrast, the constraint encoding enabled
by library-aware solving essentially captures the semantics of the
program without explicitly handling the loop in \texttt{strlen}. 
Only one query is needed to check
whether the length $S2_{sym}$ can be smaller than the length
$S1_{sym}$.

\subsection{Experiment III: Binary Search Heuristic} 

In the case of unconstrained string variables, both Z3str2
and \toolname negotiate with the Z3 core to converge on concrete
length assignments. Z3str2 does so via a linear length search
approach. We evaluate this approach against the binary search
heuristic. For that, we have implemented a second version of \toolname
that applies linear search. We adapted benchmarks used to validate
Z3str2 \cite{z3str2Test}, resulting in a total of 109 tests, which we used to compare between the two versions. The tests make heavy use of string and bit-vector
operators. Timeout was set at 20 seconds per benchmark.  The
comparison results are presented in Table~\ref{tab:binarySearch}. We
group the instances by the solver result: \textsf{SAT}, \textsf{UNSAT},
\textsf{TIMEOUT} or \textsf{UNKNOWN}.

The \toolname solver is able to complete on all instances in $17.7$
seconds, whereas its version with linear search requires $548.1$
seconds. This version can solve simple \textsf{SAT} cases, but times
out on 26 of the harder \textsf{SAT} cases, whereas \toolname has zero
timeouts. \toolname is able to detect overlapping arrangements in 2
cases, on which it returns \textsf{UNKNOWN}. The linear-search
version, in contrast, can only complete on one of the \textsf{UNKNOWN}
instances. Notably, both solvers neither crashed nor reported any
errors on any of the instances. These results lend support to the idea
that binary search heuristic significantly faster than linear search.

\begin{table}[t]
	\centering
	\caption{Performance Comparison of Search Heuristics.}
	\label{tab:binarySearch}
	\vspace{-0.1in}
	{
		\bgroup
		\def\arraystretch{1.3}
		\resizebox{\columnwidth}{!}{
			\begin{tabular}{|c||r|r|r|r|r|r|r|r|r|r|r|r|r|r|r|r||r|c|}
				\hline
				\multirow{2}{0.2\columnwidth}{\centering \toolname}
				&  \multicolumn{4}{|c|}{SAT} 
				&  \multicolumn{4}{|c|}{UNSAT} 
				&  \multicolumn{4}{|c|}{TIMEOUT (20s)} 
				&  \multicolumn{4}{|c||}{UNKNOWN} 
				&  \multicolumn{2}{|c|}{Total} \\
				\cline{2-19}
				&  \#  & $T_{min}$ & $T_{avg}$ & $T_{max}$    
				&  \#  & $T_{min}$ & $T_{avg}$ & $T_{max}$
				&  \#  & $T_{min}$ & $T_{avg}$ & $T_{max}$
				&  \#  & $T_{min}$ & $T_{avg}$ & $T_{max}$
				&  \#  & Time(s) \\
				\hline
				Binary Search
				&  98 & 0.060 & 0.172 & 2.667
				&   9 & 0.047 & 0.081 & 0.320
				&   0 & 0 & 0 & 0
				&   2 & 0.051 & 0.085 & 0.118
				& 109 & 17.7 (\textbf{1x}) \\
				\hline
				Linear Search
				&  72 &  0.060 &  0.097 &  0.618
				&   9 &  0.061 &  0.111 &  0.415
				&  27 & 20.000 & 20.000 & 20.000
				&   1 &  0.072 &  0.072 &  0.072
				& 109 & 548.1 (\textbf{31x}) \\
				\hline
			\end{tabular}
		}
		\egroup
	}
\end{table}

\section{Related Work}

While we are unaware of existing solver engines for a \emph{combined}
QF first-order many-sorted theory of strings and bit-vectors,
considerable progress has been made in developing solvers that model
strings either natively or as bit-vectors. We survey some of the main
results in this space.

\noindent {\bf String solvers:}
Zheng et al. \cite{cav15} present a solver for the QF many-sorted
theory $T_{wlr}$ over word equations, membership predicate over
regular expressions, and length function, which consists of the string
and numeric sorts.  The solver algorithm features two main heuristics:
(i) sound pruning of arrangements with overlap between variables,
which guarantees termination, and (ii) bi-directional integration
between the string and integer theories.  S3 \cite{s3} is another
solver with similar capabilities. S3 reuses Z3str's word-equation
solver, and handles regex membership predicates via
unrolling. CVC4 \cite{CVC4-CAV14} handles constraints over the theory
of unbounded strings with length and regex membership. It is based on
multi-theory reasoning backed by the DPLL($T$) architecture combined
with existing SMT theories.  The Kleene operator in regex membership
formulas is dealt with via unrolling as in Z3str2. Unlike \toolname,
these techniques all model string length as an integer, which makes it
difficult to reason about potential overflow.  In particular, none of
these approaches combines strings and bit-vectors into a unified
theory.

Another approach is to represent string variables as a regular
language or a context-free grammar (CFG).  JSA \cite{sas03} computes
CFGs for string variables in Java programs. Hooimeijer et
al.  \cite{ase10_weimer} suggest an optimization, whereby automata are
built lazily.  Other heuristics, to eliminate inconsistencies, are
introduced as part of the Rex algorithm~\cite{rex,rex2}. To overcome
the challenge faced by automata-based approaches of capturing
connections between strings and other domains (e.g. to model string
length), refinements have been proposed.  JST \cite{JST} extends
JSA. It asserts length constraints in each automaton, and handles
numeric constraints after conversion.  PISA \cite{PISA} encodes Java
programs into M2L formulas that it discharges to the MONA solver to
obtain path- and index-sensitive string
approximations. PASS \cite{pass,SymJS} combines automata and
parameterized arrays for efficient treatment of UNSAT cases.  Stranger
extends string automata with arithmetic
automata~\cite{stranger,yu_tacas09}.  For each string automaton, an
arithmetic automaton accepts the binary representations of all
possible lengths of accepted strings. Norn~\cite{norn} relates
variables to automata. Once length constraints are addressed, a
solution is obtained by imposing the solution on variable
languages. These solutions, similarly to Z3str2, offer model string
length as an integral value, thereby failing to directly capture the
notion of overflow. The S-Looper tool \cite{slooper} addresses the
specific problem of detecting buffer overflows via summarization of
string traversal loops. The S-Looper algorithm combines static
analysis and symbolic analysis to derive a constraint system, which it
discharges to S3 to detect whether overflow conditions have been
satisfied. While S-Looper is effective, it operates under a set of
assumptions that limit its applicability (e.g. no loop nesting and
only induction variables in conditional branches). Z3str2+BV, in
contrast, is a general solution for system-level programs.

\noindent {\bf Bit-vector-based Solvers:}
Certain solvers convert string and other constraints to bit-vector
constraints.  HAMPI~\cite{hampi} is an efficient solver for string
constraints, though it requires the user to provide an upper bound on
string lengths. The bit-vector constraints that it generates are
discharged to STP~\cite{stp}.  Kaluza~\cite{kaluza} extends both STP
and HAMPI to support mixed string and numeric constraints. It
iteratively finds satisfying length solutions and converts multiple
versions of fixed-length string constraints to bit-vector problems. A
similar approach powers Pex~\cite{tacas09} to address the path
feasibility problem, though strings are reduced to integer
abstractions. The main limitation of solvers like HAMPI is the
requirement to bound string lengths.  In our approach, there is no
such limitation.

\section{Conclusion and Future Work}

We have presented \toolname, a solver for a combined quantifier-free
first-order many-sorted theory of string equations, string length, and
linear arithmetic over bit-vectors. This theory has the necessary
expressive power to capture machine-level representation of strings and
string lengths, including the potential for overflow. We motivate the
need for such a theory and solver by demonstrating our ability to
reproduce known buffer-overflow vulnerabilities in real-world
system-level software written in C/C++. We also establish a foundation
for unified reasoning about string and bit-vector constraints in the
form of a decidability result for the combined theory.

\newpage

\bibliographystyle{abbrv}
\bibliography{main,cav15,uw-ethesis}

\appendix

\newpage

\begin{subappendices}

\renewcommand{\thesection}{\Alph{section}}
  
\section{Proof of Lemma~\ref{lem:bv2regex}} \label{proof:lem:bv2regex}

We wish to establish a conversion from bit-vector constraints to
regular languages. For regular expressions (regexes), we use the
following standard notation: $AB$ denotes the concatenation of regular
languages $A$ and $B$.  $A|B$ denotes the alternation (or union) of
regular languages $A$ and $B$.  $A^{*}$ denotes the Kleene closure of
regular language $A$ (i.e., 0 or more occurrences of a string in $A$).
For a finite alphabet $\Sigma = \{a_1, a_2, \hdots, a_l\}$, $\left[
a_1 - a_l \right]$ denotes the union of regex $a_1 | a_2 | \hdots |
a_l$.  Finally, $A^{i}$, for nonzero integer constants $i$ and regex
$A$, denotes the expression $A A \hdots A$, where the term $A$ appears
$i$ times in total.

\vspace{0.2cm}
\noindent{\bf Lemma~\ref{lem:bv2regex}}:
Let $k$ be the width of all bit-vector terms.  Suppose we have a
bit-vector formula of the form $len_{bv}(X) = C$, where $X$ is a
string variable and $C$ is a bit-vector constant of width $k$.  Let
$i_{C}$ be the integer representation of the constant $C$,
interpreting $C$ as an unsigned integer.  Then the set $M(X)$ of all
strings satisfying this constraint is equal to the language $L$
described by the regular expression $(\left[a_1 -
a_l\right]^{2^{k}})^{*} \left[a_1 - a_l\right]^{i_{C}}$.

\begin{proof}
  In the forward direction, we show that $M(X) \subseteq L$. Let
  $x \in M(X)$.  $x$ satisfies the constraint $len_{bv}(x) = C$, which
  means that the integer length $z$ of $x$ modulo $2^k$ is equal to
  $i_{C}$.  Additionally, $z \ge 0$ as strings cannot have negative
  length. Then there exists a non-negative integer $n$ such that $z = n
  2^{k} + i_{C}$. We decompose $x$ into strings $u, v$ such that $uv =
  x$, the length of $u$ is $n 2^{k}$, and the length of $v$ is
  $i_{C}$. Now, $u \in (\left[a_1 - a_l\right]^{2^{k}})^{*}$ because
  its length is a multiple of $2^k$, and $v \in \left[a_1 -
  a_l\right]^{i_{C}}$ because its length is exactly $i_{C}$. By
  properties of regex concatenation, $uv \in (\left[a_1 -
  a_l\right]^{2^{k}})^{*} \left[a_1 - a_l\right]^{i_{C}}$ and
  therefore $uv \in L$.  Since $uv = x$, we have $x \in L$.

  In the reverse direction, we show that $L \subseteq M(X)$. Let
  $x \in L$.  By properties of regex concatenation, there exist
  strings $u, v$ such that $uv = x$, $u \in (\left[a_1 -
  a_l\right]^{2^{k}})^{*}$, and $v \in \left[a_1 -
  a_l\right]^{i_{C}}$.  Suppose $u$ was matched by $n$ expansions of
  the outer Kleene closure for some non-negative integer $n$.  Then the
  integer length of $u$ is $n 2^{k}$. Furthermore, the integer length
  of $v$ is $i_{C}$.  This implies that the integer length of $x$ is
  $n2^{k} + i_{C}$, which means that the integer length of $x$ is
  equal to $i_{C}$ modulo $2^{k}$, from which it directly follows that
  $len_{bv}(x) = C$. Hence $x \in M(X)$ as required. This completes
  both directions of the proof and so we have equality between the
  sets $M(X) = L$.
\end{proof}

\section{Proof of Theorem~\ref{thm:strbvdecidable}} \label{proof:thm:strbvdecidable}

\paragraph{Theorem~\ref{thm:strbvdecidable}} The satisfiability problem for the QF theory of word equations and
  bit-vectors is decidable.

\begin{proof}
  We demonstrate a decision procedure by reducing the input formula to
  a finite disjunction of subproblems in the theory of QF word
  equations and regular language constraints.  This theory is known to
  be decidable by Schulz's extension of Makanin's algorithm for solving word
  equations \cite{Schulz:1990:MAW:646900.710169}.

  Suppose the input formula $\phi$ has the form $W_1 = W_2 \land A_1 =
  B_1 \land A_2 = B_2 \land \hdots \land A_n = B_n$, where $W_1, W_2$
  are terms in the theory of word equations and $A_1 \hdots A_n,
  B_1 \hdots B_n$ are terms in the theory of bit-vectors.  Let $k$ be
  the width of all bit-vector terms.  For each term of the form
  $len_{bv}(X_i)$ in $A_1 \hdots A_n, B_1 \hdots B_n$, replace it with
  a fresh bit-vector variable $v_i$ and collect the pair $(v_i, X_i)$
  in a set $\mathcal{S}$ of substitutions.  Suppose there are $m$ such
  pairs. Then the total number of bits among all variables introduced
  this way is $mk$.  This means that there are $2^{mk}$ possibilities
  for the values of $v_1 \hdots v_m$.  Because the theory of QF
  bit-vectors is decidable and because there are finitely many
  possible values for $v_1 \hdots v_m$, we can check the
  satisfiability of the bit-vector fragment of the input formula $A_1
  = B_1 \land A_2 = B_2 \land \hdots \land A_n = B_n$ for all possible
  substitutions of values for $v_1 \hdots v_m$ in finite time. For
  each assignment $A = \{ (v_1, C_1), (v_2, C_2), \hdots, (v_m,
  C_m) \}$, where each $v_i$ is a variable and each $C_i$ is a
  bit-vector constant, if the bit-vector constraints are satisfiable
  under that assignment, collect $A$ in the set $\mathcal{A}$ of all
  satisfying assignments.  If the set $\mathcal{A}$ is empty, then the
  bit-vector constraints were not satisfiable under any assignment to
  $v_1 \hdots v_m$.  In this case we terminate immediately and decide
  that the input formula is UNSAT, as the bit-vector constraints must
  be satisfied for satisfiability of the whole formula. Otherwise,
  construct the formula $R'(\phi)$ as follows.  For each assignment
  $A \in \mathcal{A}$, for each term $(v_i, C_i) \in A$, we find the
  pair $(v_i, X_i) \in \mathcal{S}$ with corresponding $v_i$.  Because
  each variable $v_i$ corresponds to a term $len_{bv}(X_i)$, and since
  we have $v_i = C_i$, we have the constraint $len_{bv}(X_i) =
  C_i$. This allows us to apply Lemma~\ref{lem:bv2regex} and generate
  a regular language constraint $X_i \in L_i$.  After generating each
  such regular language constraint, we collect $R'(\phi) :=
  R(\phi) \lor (W_1 = W_2 \land X_1 \in L_1 \land X_2 \in
  L_2 \land \hdots \land X_m \in L_m)$.  We repeat this for each
  assignment $A \in \mathcal{A}$.  The resulting formula $R(\phi) =
  (W_1 = W_2) \land R'(\phi)$ is a conjunction of the original word
  equation from $\phi$ and a finite disjunction of regular-language
  constraints over variables in that word equation.  We now invoke
  Schulz's algorithm to solve this formula. If the word equation and
  any disjunct are satisfiable, then we report that the original
  formula $\phi$ is SAT; otherwise, $\phi$ is UNSAT. Finally, it is
  easy to show that the reduction is sound, complete, and terminating
  for all inputs (see Appendix~\ref{app:reduction} for soundness and completeness proof
of the reduction).
\end{proof}

\section{Proof of Soundness and Completeness of the Reduction used in Theorem ~\ref{thm:strbvdecidable}}\label{app:reduction}

 We demonstrate that the reduction from bit-vector constraints to
 regular language constraints, as performed in the proof for
 Theorem \ref{thm:strbvdecidable}, is sound and complete. We do so by
 showing equisatisfiability between $\phi$ and $R(\phi)$.

 \begin{theorem} \label{thm:strbvequisat}
   $\phi$ is satisfiable iff $R(\phi)$ is satisfiable.
 \end{theorem}

 \begin{proof} In the forward direction, we show that if $\phi$ is
   satisfiable then $R(\phi)$ is satisfiable.  Let $M$ be a satisfying
   assignment of all variables in $\phi$.  Because $\phi$ and
   $R(\phi)$ share the same constraint $W_1 = W_2$, $M$ is a
   satisfying assignment for the word equation fragment of $R(\phi)$
   as well.  It remains to show that at least one of the terms in
   $R'(\phi)$, the disjunction of regular language constraints, is
   satisfiable. The algorithm described in
   Theorem~\ref{thm:strbvdecidable} generates one group of regex
   constraints for each satisfying assignment to the bit-vector
   fragment that produces a distinct model for all bit-vector length
   constraints. In particular, the algorithm generates
   regular language constraints for the particular model described by
   $M$ of bit-vector length constraints.  Because the string variables
   in $M$ satisfy these constraints, we apply Lemma~\ref{lem:bv2regex}
   to find that the regular language constraints that were generated
   with respect to this model $M$ are satisfied by the assignment of
   all string variables in $M$.  Therefore $M$ is also a model of
   $R(\phi)$, and hence $R(\phi)$ is satisfiable.

   In the reverse direction, we show that if $R(\phi)$ is satisfiable
   then $\phi$ is satisfiable.  Let $M$ be a satisfying assignment of
   all variables in $R(\phi)$.  Because $R(\phi)$ and $\phi$ share the
   same constraint $W_1 = W_2$, $M$ is a satisfying assignment for the
   word-equation fragment of $\phi$ as well. It remains to show that
   the bit-vector constraints in $\phi$ are satisfiable under this
   assignment to the string variables. Let $r$ be a regular-language
   constraint in $R'(\phi)$ (the disjunction of regular-language
   constraints), such that $r$ evaluates to true under the assignment
   $M$. We know that such an $r$ must exist because the formula
   $R(\phi)$ is satisfiable, and therefore at least one of the terms
   in the disjunction $R'(\phi)$ must evaluate to true. By applying
   Lemma~\ref{lem:bv2regex} ``backwards'', we can derive an assignment
   of constants to bit-vector length terms in $\phi$ corresponding to
   each regular-language constraint in $r$ that is consistent with the
   lengths of the string variables. We also know that the bit-vector
   constraints are satisfiable under this assignment of constants to
   strings and bit-vector length terms because, by
   Lemma~\ref{lem:bv2regex}, a precondition for the appearance of any
   term $r$ in $R'(\phi)$ is that the bit-vector fragment of $\phi$ is
   satisfiable under the partial assignment to bit-vector length terms
   that yielded $r$. Therefore, by solving the remaining bit-vector
   constraints, which must be satisfiable, $M$ can be extended to a
   model of $\phi$ and hence $\phi$ is satisfiable.  \end{proof}
\end{subappendices}

\section{Proof of the Soundness of Algorithm~\ref{alg:highLevel}} \label{proof:alg:highLevel}

We use the standard definition of soundness for decision procedures
from the SMT literature \cite{cav15}, whereby a solver is sound if
whenever the solver returns UNSAT, the input formula is indeed
unsatisfiable.

\vspace{-0.05in}
\begin{theorem} 
	Algorithm~\ref{alg:highLevel} is sound, i.e., when
	Algorithm~\ref{alg:highLevel} reports UNSAT, the input
	constraint is indeed UNSAT.
\end{theorem} 
\vspace{-0.05in}

\begin{proof}
First, line \ref{alg:unsat1} returns an UNSAT if
either the string or the bit-vector constraints are determined to be UNSAT. For
string constraints, we use the algorithms described in \cite{ganesh2012} to decide the satisfiability of word (dis)equations. Soundness follows from
the soundness of the procedure \cite{ganesh2012} and the (Z3) bit-vector solver.

For the UNSAT returned at line \ref{alg:unsat2}, we show
that transformations impacting it are all satisfiability-preserving. 
In particular, the transformations at $(i)$ line
\ref{alg:conv1} $(ii)$ lines \ref{alg:conv2_s}-\ref{alg:conv2_e} $(iii)$ line
\ref{alg:merge_s}-\ref{alg:merge_e} and $(iv)$ line \ref{alg:conv3} are
satisfiability-preserving: 
$(i)$ The DNF conversion at line \ref{alg:conv1} is 
obviously satisfiability-preserving. 
$(ii)$ Line \ref{alg:conv2_s} is a variant of the sound arrangement generation 
method mentioned in Makanin's paper \cite{makanin}. It is 
satisfiability-preserving because each arrangement is a finite set of equations 
implied by the input system of equations. Besides, we extract length constraints 
from arrangements. If they conflict with the existing bit-vector constraints, we 
drop the corresponding arrangements. As we assume the bit-vector theory is sound, 
this step is also satisfiability-preserving. 
$(iii)$ Lines \ref{alg:merge_s}-\ref{alg:merge_e} systematically enumerate all feasible options to further split word equations \cite{cav15}. This step
is satisfiability-preserving. $(iv)$ Line \ref{alg:conv3} derives simpler equations by a satisfiability-preserving rewriting \cite{cav15}.
\end{proof}

\end{document}